\documentclass[lettersize,journal]{IEEEtran}
\usepackage{amsmath,amsfonts}
\usepackage{algorithm}
\usepackage{algpseudocode} 
\usepackage{algorithm}
\usepackage{array}
\usepackage{makecell}
\usepackage{booktabs}
\usepackage[caption=false,font=footnotesize,labelfont=sf,textfont=sf]{subfig}
\usepackage{textcomp}
\usepackage{url}
\usepackage{verbatim}
\usepackage{graphicx}
\usepackage{cite}

\newtheorem{theorem}{Theorem}
\newtheorem{lemma}{Lemma}
\newtheorem{definition}{Definition}
\begin{document}

\title{A two-steps tensor eigenvector centrality for nodes and hyperedges in hypergraphs}

\author{Qing~Xu,
        Chunmeng~Liu,
        Changjiang~Bu,
        and~Jihong~Shen
\thanks{Qing Xu, Changjiang Bu, and Jihong Shen are with the School of Mathematical Sciences, Harbin Engineering University, Harbin 150001, PR China (e-mail: qinxu191@163.com; buchangjiang@hrbeu.edu.cn; shenjihong@hrbeu.edu.cn).}
\thanks{Chunmeng Liu is with the Academy for Advanced Interdisciplinary Studies, Northeast Normal University, Changchun 130024, PR China (e-mail: liuchunmeng0214@126.com).}
\thanks{Corresponding author: Chunmeng Liu.}
}

\markboth{Journal of \LaTeX\ Class Files,~Vol.~XX, No.~X, August~202X}%
{Xu \MakeLowercase{\textit{et al.}}: Title of Your Paper}


\markboth{Journal of \LaTeX\ Class Files,~Vol.~14, No.~8, August~2021}%
{Shell \MakeLowercase{\textit{et al.}}: A Sample Article Using IEEEtran.cls for IEEE Journals}


\maketitle

\begin{abstract}
Hypergraphs have been a powerful tool to represent higher-order interactions, where hyperedges can connect an arbitrary number of nodes. Quantifying the relative importance of nodes and hyperedges in hypergraphs is a fundamental problem in network analysis. In this paper, we propose a new tensor-based centrality measure for general hypergraphs. We use a third-order tensor to represent the relationship between nodes and hyperedges. The tensor's positive Perron vector is defined as the centrality vector of the hypergraph. The existence and uniqueness of this centrality vector are guaranteed by the Perron-Frobenius theorem for tensors. This new centrality measure captures a higher-order mutual reinforcement mechanism: a node's importance is determined by the importance of its incident hyperedges and the other nodes within these hyperedges; symmetrically, a hyperedge's importance is determined by the importance of its constituent nodes and the other hyperedges containing these nodes. We further provide a combinatorial interpretation by proving that the centrality vector represents the limit geometric capacity of two-steps expansion trees. We illustrate the centrality measure on real-world hypergraph datasets.\end{abstract}


\section{Introduction}
\IEEEPARstart{M}{odeling} complex systems often requires capturing interactions that go beyond simple pairwise connections. Hypergraphs generalize the concept of simple graphs and have emerged as a powerful tool to represent such higher-order interactions. Let $H = (\mathcal{V}, \mathcal{E})$ be a hypergraph, where $\mathcal{V}$ is the set of nodes, and $\mathcal{E}$ is the set of hyperedges. A hyperedge $e \in \mathcal{E}$ can contain an arbitrary number of nodes, while an edge in a graph can only connect two nodes. If each hyperedge contains exactly $k$ nodes, then $H$ is referred to as a $k$-uniform hypergraph; otherwise, $H$ is called non-uniform \cite{Berge_1984}. Given a hypergraph, a fundamental issue is evaluating the centrality or relative importance of individual nodes or hyperedges. In the analysis of standard graphs, this issue has been extensively studied using various centrality measures, including degree, closeness, betweenness, Katz, subgraph, and eigenvector centrality \cite{Lu_2016,Ernesto_2005,Ernesto_2006,Bonacich_1987,Bonacich_2007}. Building on these established foundations, some centrality measures have been generalized to hypergraphs to quantify the significance of nodes and hyperedges \cite{Lu_2025}.

Among these measures, eigenvector centrality is an important centrality measure with widespread applications in social network analysis \cite{Fan_2020} and web search ranking \cite{Gleich_2015}. Inspired by eigenvector centrality in simple graphs, researchers have extended it to hypergraphs. In 2019, Benson extended eigenvector centrality to $k$-uniform hypergraphs \cite{Benson_2019}. This approach is analogous to the graph case: just as graph eigenvector centrality corresponds to the Perron vector of the adjacency matrix, uniform hypergraph centrality is derived from the Perron vector of the adjacency tensor. However,  real-world hypergraphs are often non-uniform. For non-uniform hypergraphs, Tudisco and Higham formulated hypergraph centrality as a general constrained nonlinear eigenvalue problem, defining a class of spectral measures to identify important nodes and hyperedges in hypergraphs \cite{Tudisco_2021}. For more research on hypergraph centrality, please refer to Refs. \cite{Ernesto_2006, Kovalenko_2022,xie2023vital,Clark_2023}.

In this paper, we propose a tensor-based centrality measure for general hypergraphs. We first transform the hypergraph into its incidence bipartite graph and construct a third-order tensor to encode the two-steps walks, which was originally proposed in Ref. \cite{Xu_2023}. Specifically, the tensor entry $a_{ijk}=1$ indicates a walk $i \sim j \sim k$ in the incidence bipartite graph, representing either two nodes $i, k$ sharing a hyperedge $j$, or two hyperedges $i, k$ intersecting at a node $j$. Assuming that the importance of a node or hyperedge is proportional to the joint influence of the connecting hyperedge and the neighbor (quantified as the product of their scores), we derive the centrality vector as the unique positive eigenvector corresponding to the spectral radius of the tensor. The existence and uniqueness of this positive vector are guaranteed by the Perron-Frobenius theorem for tensors. We use an iterative algorithm based on the power method for nonnegative tensors to compute the centrality vector efficiently \cite{Ng_2010,Liu_2010,Zhou_2013}. Furthermore, we interpret the centrality using two-steps expansion trees. We define a recursive metric called geometric capacity on these trees and prove that the centrality vector is exactly the normalized limit of these capacities as the tree depth approaches infinity.

The remainder of this paper is organized as follows. Section 2 introduces definitions and fundamental lemmas related to tensors. Section 3 details the proposed tensor-based centrality measure and provides an algorithm to compute the centrality vector. Section 4 presents numerical experiments on real-world datasets to illustrate the effectiveness of the centrality measure. Finally, Section 5 gives the conclusions of the paper.

\section{Preliminaries}
In this section, we review fundamental concepts regarding tensors and their spectral properties, which serve as the mathematical foundation for our proposed method.

Let $\mathbb{C}^{n}$ denote the set of $n$-dimensional vectors over the complex field $\mathbb{C}$.
 We consider an $m$-order $n$-dimensional tensor $\mathcal{A} = (a_{i_{1}i_{2}\cdots i_{m}})$. The tensor $\mathcal{A}$ is called a \emph{nonnegative tensor} if all its entries are nonnegative. For a vector $\mathbf{x} = (x_{1}, \ldots, x_{n})^{\mathrm{T}} \in \mathbb{C}^{n}$, the product $\mathcal{A}\mathbf{x}^{m-1}$ is a vector in $\mathbb{C}^{n}$ whose $i$-th component is given by $\sum_{i_{2},\ldots,i_{m}=1}^{n}a_{ii_{2}\cdots i_{m}}x_{i_{2}}\cdots x_{i_{m}}$.
A complex number $\lambda$ is called an \emph{eigenvalue} of $\mathcal{A}$ if there exists a nonzero vector $\mathbf{x} \in \mathbb{C}^{n}$ such that
\begin{equation*}
  \mathcal{A} \mathbf{x}^{m-1} = \lambda \mathbf{x}^{[m-1]},
\end{equation*}
where $\mathbf{x}^{[m-1]} = (x_1^{m-1}, \ldots, x_n^{m-1})^{\mathrm{T}}$. The vector $\mathbf{x}$ is then called an \emph{eigenvector} of $\mathcal{A}$ corresponding to $\lambda$ \cite{Qi_2005,Lim_2005}.
The spectral radius of $\mathcal{A}$ is the largest modulus of its eigenvalues, denoted by $\rho(\mathcal{A})$. These spectral concepts have been widely applied in spectral graph theory, spectral hypergraph theory and network analysis, see Refs. \cite{Liu_2023,app_4,app_7,app_8,Liu_2023_2,Liu_2024,Liu_2026,Tudisco_2021,Clark_2023,Xu_2023,Wu_2019,Zhang_2024}.

A crucial property for the existence of a unique positive eigenvector is the tensor's weakly irreducibility, which can be characterized via the tensor's representative matrix. 
\begin{definition}{\rm \cite{Qi_2017}}
Let $\mathcal{A}=(a_{i_{1}i_{2}\cdots i_{m}})$ be an $m$-th order $n$-dimensional tensor. The \textit{representative matrix} of $\mathcal{A}$, denoted by $M(\mathcal{A}) = (m_{ij})$, is defined by:
\begin{equation*}
    m_{ij} = \sum_{j \in \{i_2, \dots, i_m\}} |a_{i i_2 \cdots i_m}|.
\end{equation*}
\end{definition}

The tensor $\mathcal{A}$ is said to be \emph{weakly irreducible} if and only if its representative matrix $M$ is irreducible \cite{Qi_2017}. A weakly irreducible nonnegative
tensor is called a \textit{weakly primitive} tensor if its representative matrix is primitive \cite{Qi_2017}.
For weakly irreducible nonnegative tensors, the following Perron-Frobenius theorem holds.

\begin{lemma}{\rm \cite{Perron Th}} \label{thm Perron Th.}
If $\mathcal{A}$ is a nonnegative weakly irreducible tensor, then its spectral radius $\rho(\mathcal{A})$ is the unique positive eigenvalue of $\mathcal{A}$, and there exists a unique positive eigenvector $\mathbf{x}$ (up to a positive scaling factor) corresponding to $\rho(\mathcal{A})$.
\end{lemma}

Our work builds upon the concept of the two-steps tensor, originally introduced in Ref. \cite{Xu_2023} as follows.
\begin{definition}{\rm \cite{Xu_2023}}
Let $G = (V,E)$ be a simple graph with $n$ nodes.
The two-steps tensor $\mathcal{A}(G) = (a_{ijk})$ is a third-order $n$-dimension tensor defined by
\[
a_{ijk} =
\begin{cases}
1, & \text{\rm{if} $\{i,j\}\in E,\{j,k\} \in E$,} \\
0, & \text{\rm{otherwise}}.
\end{cases}
\]
The spectral radius $\rho(G)$ is the largest modulus among all eigenvalues of $\mathcal{A}(G)$.
\end{definition}

The connectivity of the graph ensures the weakly irreducibility of this tensor, as stated in the following lemma.
\begin{lemma}\textup{\cite{Xu_2023}}\label{lemma2.4}
  Let $G = (V,E)$ be a simple undirected connected graph. Then the two-steps tensor associated to $G$ is nonnegative weakly irreducible. 
\end{lemma}

In this paper, we strengthen Lemma \ref{lemma2.4} by demonstrating that the tensors are weakly  primitive.
\begin{lemma}\label{lemma2.5}
Let $G = (V,E)$ be a simple undirected connected graph. The associated two-steps tensor $\mathcal{A} = (a_{ijk})$ is weakly primitive.
\end{lemma}

\begin{IEEEproof}
By Lemma \ref{lemma2.4}, the tensor $\mathcal{A}$ is weakly irreducible, which implies that its representative matrix $M(\mathcal{A})$ is irreducible. To prove that $\mathcal{A}$ is weakly primitive, it suffices to show that $M(\mathcal{A})$ is primitive.

According to the definition of the representative matrix, the diagonal entry is given by $m_{ii} = \sum_{k} (a_{iki} + a_{iik})$. Since the graph contains no isolated nodes, every node $i$ has at least one neighbor $k$. This guarantees the existence of the backtracking walk $i \sim k \sim i$, which implies $a_{iki} = 1$. Consequently, for each $i$, the diagonal entries satisfy:
\begin{equation*}
    m_{ii} \ge a_{iki} = 1 > 0, 
\end{equation*}

It is a classical result in matrix theory that an irreducible nonnegative matrix with a strictly positive diagonal is primitive. Since $M(\mathcal{A})$ satisfies these conditions, it is primitive. Therefore, the tensor $\mathcal{A}$ is weakly primitive.
\end{IEEEproof}

\section{The two-steps eigenvector centrality of general hypergraphs}
\subsection{The definition of hypergraph two-steps eigenvector centrality (HTEC)}
In this section, we propose a novel tensor-based centrality measure, the hypergraph two-steps eigenvector centrality (HTEC).

A hypergraph $H = (\mathcal{V}, \mathcal{E})$ can be represented by its incidence bipartite graph, denoted as $B(H)$. The node set of $B(H)$ is constructed as the union of the original nodes and the hyperedges. In this structure, an undirected edge connects a node to a hyperedge if and only if the node is a member of that hyperedge, as illustrated in Fig. 1. If $H$ is connected, then $B(H)$ is also connected. This structural property allows us to directly apply the proposed two-steps tensor model to define centrality in hypergraphs. We define the hypergraph centrality measure as follows, and refer to it as the hypergraph two-steps eigenvector centrality (HTEC). 

\begin{figure}[!t]
\centering
\includegraphics[width=\columnwidth]{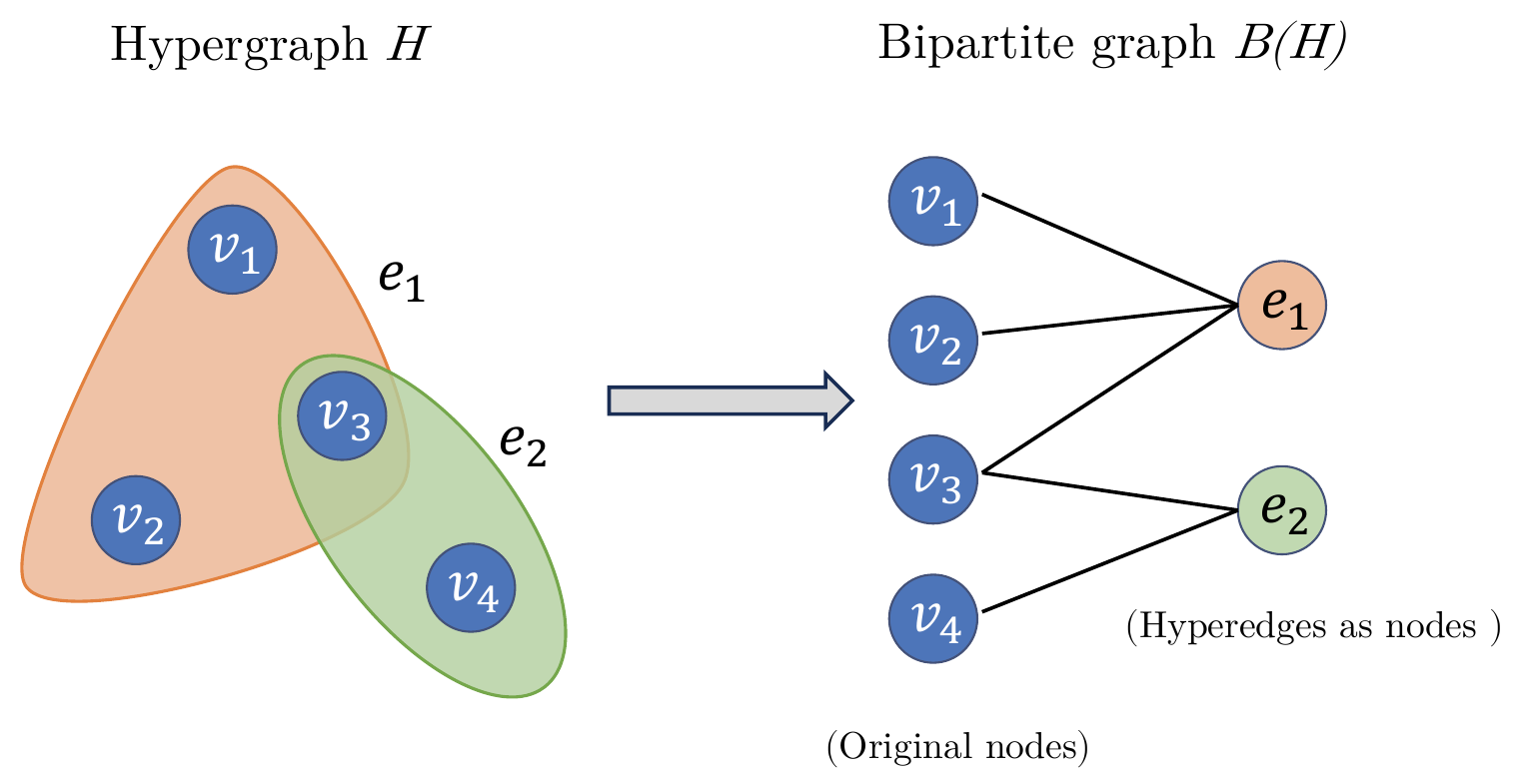}
\caption{A schematic diagram of converting a hypergraph to a bipartite graph.}
\label{hyper2bi}
\end{figure}

\begin{definition}
Let $H = (\mathcal{V}, \mathcal{E})$ be a connected hypergraph. Let $B(H) = (V_B, E_B)$ denote the incidence bipartite graph of $H$ with the node set $V_B = \mathcal{V} \cup \mathcal{E}$ of size $n = |\mathcal{V}| + |\mathcal{E}|$. Let $\mathcal{A}(B)$ be the two-steps tensor associated with $B(H)$.
The hypergraph two-steps centrality vector $\mathbf{x} \in \mathbb{R}^{n}$ is defined as the unique positive eigenvector of $\mathcal{A}$, satisfying the equation:
\begin{equation}\label{eq:vector_form}
\mathcal{A}(B)\mathbf{x}^2 = \rho(B) \mathbf{x}^{[2]},
\end{equation}
where $\rho(B)$ is the spectral radius of $\mathcal{A}(B)$.
\end{definition}

To interpret the physical meaning of this measure, we express Eq.\eqref{eq:vector_form} in its element-wise form. For any node $i \in V_B$ (representing either an original node or a hyperedge), the centrality score satisfies:
\begin{equation}\label{eq:hyper_tensor_eq}
\rho(B) x_i^2 = \sum_{j=1}^n \sum_{k=1}^n a_{ijk} x_j x_k.
\end{equation}
For the incidence graph $B(H)$, a non-zero tensor entry $a_{ijk}=1$ corresponds to a two-steps walk $i \sim j \sim k$. Since the node set of $B(H)$ consists of both original nodes $\mathcal{V}$ and hyperedges $\mathcal{E}$, the interpretation of this walk depends on whether the node $i$ represents an original node or a hyperedge. 

\textbf{Case 1:} If $i$ is a node $v \in \mathcal{V}$, the two-steps walk takes the form $v \sim e \sim u$, where $e \in \mathcal{E}$ is a hyperedge and $u \in \mathcal{V}$ is a node. Here, $a_{veu}=1$ represents a co-membership relationship through a shared hyperedge $e$. By substituting $a_{veu}=1$ into Eq.\eqref{eq:hyper_tensor_eq}, we see that the influence of node $i$ is quantified by the product of the scores of the hyperedge and the neighbor ($x_e x_u$). Consequently, a node $v$ achieves high centrality not merely by belonging to many hyperedges, but by sharing influential hyperedges ($x_e$ is high) with other influential nodes ($x_u$ is high). For example, in a scientific collaboration network, a researcher becomes central not just by publishing in top journals (high hyperedge quality), but by collaborating with distinguished  scholars (high node quality). 

\textbf{Case 2:} Symmetrically, if $i$ is a hyperedge $e \in \mathcal{E}$, the two-steps walk takes the form $e \sim v \sim f$, where $v \in \mathcal{V}$ is a common node and $f \in \mathcal{E}$ is a hyperedge. In this situation, a hyperedge acquires high centrality if it shares nodes with other influential hyperedges. For instance, a research team is ranked highly if its members are also collaborating with other prestigious teams.

It is natural to compare our method with the standard eigenvector centrality applied to the incidence bipartite graph $B(H)$. In element-wise form, the standard eigenvector centrality of a node $i$ is determined by the linear sum of its neighbors' scores. In the context of the hypergraph, this implies that a node's importance is the sum of the importance of the hyperedges it belongs to, and a hyperedge's importance is the sum of the importance of its constituent nodes. In contrast, our proposed model $\rho x_i^2 = \sum_{j,k=1}^n a_{ijk} x_j x_k$ relies on a product term, where the contribution of each two-steps walk is determined by the product of the hyperedge score and the neighbor score. Therefore, unlike the standard eigenvector centrality model, which emphasizes the quantity of the neighbors, the proposed model emphasizes the joint influence of the connecting hyperedges and the neighbors.

\subsection{Combinatorial interpretation: geometric capacity of two-steps expansion trees}
It is a classical result in spectral graph theory that standard eigenvector centrality is related to the limit distribution of walk counts. Analogously, in this part, we provide a combinatorial interpretation for the proposed HTEC measure. Specifically, we demonstrate that HTEC characterizes the limit geometric capacity of two-steps expansion trees.

We begin by revisiting the combinatorial interpretation of standard eigenvector centrality and then extend it to our method. We define the \textit{linear expansion tree} and its associated \textit{linear capacity} as they are important for the following analysis.

\begin{definition}
Let $G$ be a simple graph with adjacency matrix $A = (a_{ij})$. For any node $i$, the \textit{linear expansion tree} of depth $t$, denoted $\mathcal{L}_t(i)$, is defined recursively:
\begin{enumerate}
\item Base case ($t=0$): $\mathcal{L}_0(i)$ consists of a single node $i$ (the root).
\item Recursive step ($t > 0$): The root of $\mathcal{L}_t(i)$ is $i$. For every neighbor $j$ connected to $i$ (i.e., every walk $i \sim j$) in $G$, the root has one sub-tree as a child: the tree $\mathcal{L}_{t-1}(j)$ rooted at node $j$ with depth $t-1$.
\end{enumerate}
\end{definition}

A schematic illustration of this linear expansion tree is provided in the middle panel of Fig. \ref{expansion}. Based on this recursive structure, we define the linear capacity of the expansion tree.

\begin{figure*}[!t]
\centering
\includegraphics[width=\textwidth]{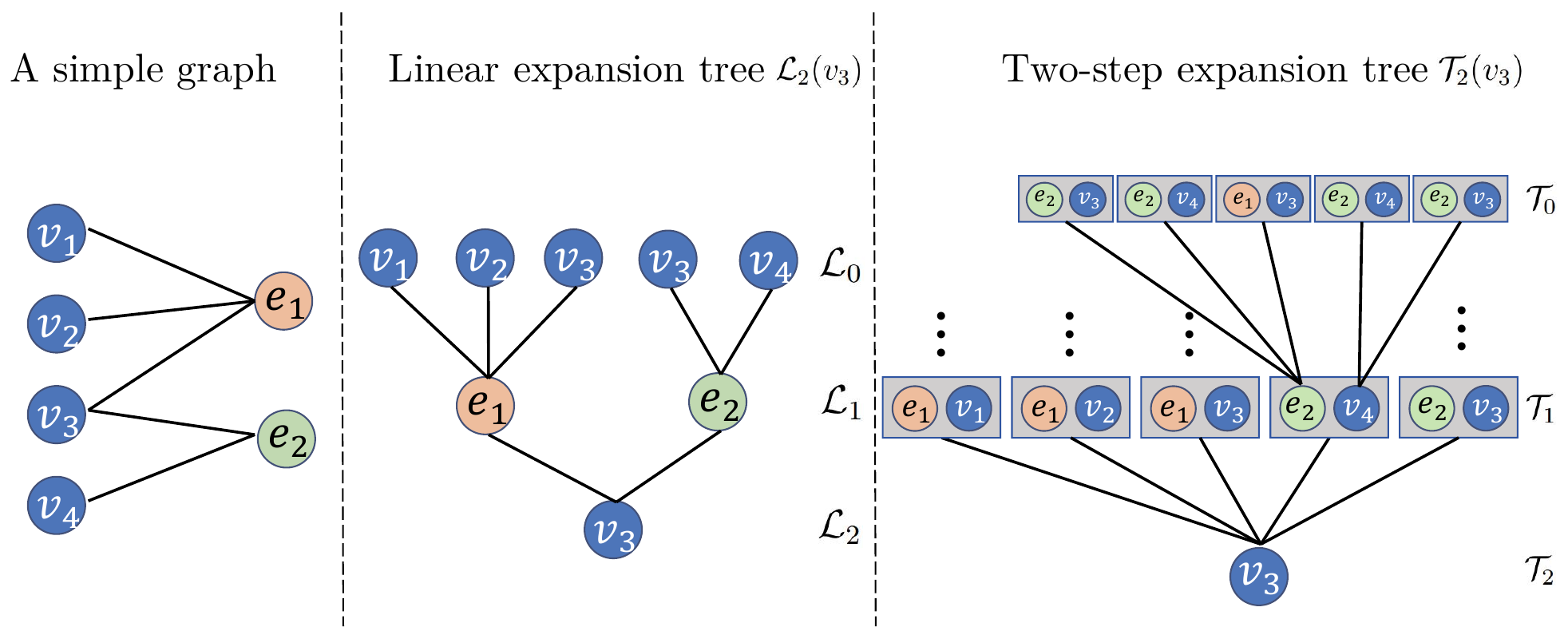}
\caption{Illustration and comparison of expansion trees rooted at node $v_3$ with depth $t=2$. The left panel shows the simple graph. The middle panel depicts the linear expansion tree $\mathcal{L}_2(v_3)$. The right panel depicts the two-steps expansion tree $\mathcal{T}_2(v_3)$.}
\label{expansion}
\end{figure*}

\begin{definition}
The \textit{linear capacity} of a linear expansion tree, denoted by $L_t(i)$, is the quantitative value associated with the expansion tree $\mathcal{L}_t(i)$. It is defined by the following recursive formula:
\begin{equation}
    L_t(i) = \sum_{j=1}^n a_{ij} L_{t-1}(j), \quad \text{for } t > 0,
\end{equation}
where $a_{ij} = 1$ if there exists an edge $i \sim j$ in $G$, and $a_{ij} = 0$ otherwise. The initial condition is $L_0(i) = 1$ for all $i$.
\end{definition}

As we can see that the relationship is linear: the value at the current node is the direct sum of its neighbors' values. It is a well-known result that as the length $t$ grows to infinity, the normalized vector of these values converges to the standard eigenvector centrality for non-bipartite graphs \cite{Cvetkovic_1997}. Thus, the standard eigenvector centrality can be viewed as the limit of this linear capacity.

We now extend the above idea to our method. Just as the linear expansion tree models the one-step neighbor relationship, we need to model the two-steps neighbor relationship in the bipartite representation of hypergraphs. Specifically, a two-steps walk $i \sim j \sim k$ implies that node $i$ is determined by combining node $j$ and $k$. This leads to the following recursive definition of two-steps expansion tree. A schematic illustration of this linear expansion tree is provided in the right panel of Fig. \ref{expansion}.
\begin{definition}
Let $B(H)$ be the incidence bipartite graph. For any node $i$, the two-steps expansion tree of depth $t$, denoted $\mathcal{T}_t(i)$, is defined recursively:
\begin{enumerate}
    \item Base case ($t=0$): $\mathcal{T}_0(i)$ consists of a single node $i$ (the root).
    \item Recursive step ($t > 0$): The root of $\mathcal{T}_t(i)$ is $i$. For every two-steps walk $i \sim j \sim k$ in $B(H)$ starting from node $i$, the root has two sub-trees as children: a left sub-tree $\mathcal{T}_{t-1}(j)$ rooted at node $j$ with depth $t-1$, and a right sub-tree $\mathcal{T}_{t-1}(k)$ rooted at node $k$ with depth $t-1$.
\end{enumerate}
\end{definition}

This tree structure illustrates that the root connects to the pair of nodes $j$ and $k$. Based on this structure, we define a measure to quantify its capacity.

\begin{definition}
The \textit{geometric capacity} of a two-steps expansion tree, denoted by $C_t(i)$, is the quantitative value associated with the expansion tree $\mathcal{T}_t(i)$. It is defined by the following recursive formula:
\begin{equation}\label{eq:geometric_recurrence}
    C_t(i) = \sqrt{\sum_{j=1}^n \sum_{k=1}^n a_{ijk} C_{t-1}(j) C_{t-1}(k)}, \quad \text{for } t > 0,
\end{equation}
where $a_{ijk} = 1$ if there exists a walk $i \sim j \sim k$ in $B(H)$, and $a_{ijk} = 0$ otherwise. The initial condition is $C_0(i) = 1$ for all $i$.
\end{definition}

The product term $C_{t-1}(j)C_{t-1}(k)$ indicates that the capacity of the root depends on the joint contribution of the paired branches. 

Having defined the geometric capacity for a two-steps expansion tree with depth $t$, a natural question arises: what is the limit of these capacities as the tree grows indefinitely? The following theorem establishes that the normalized vector of these capacities converges to our proposed HTEC centrality.

\begin{theorem}\label{thm4.3}
Let $\mathbf{C}_t = (C_t(1), \dots, C_t(n))^{\top}$ be the vector of geometric capacity for a two-steps expansion tree with depth $t$. Let $\mathbf{u}_t$ be the normalized capacity vector given by:
\begin{equation}
    \mathbf{u}_t = \frac{\mathbf{C}_t}{\|\mathbf{C}_t\|}.
\end{equation}
Assume the two-steps tensor $\mathcal{A}$ is nonnegative, weakly primitive. As the depth $t \to \infty$, the sequence $\{\mathbf{u}_t\}$ converges to the unique positive HTEC vector $\mathbf{x}$ satisfying the eigenvalue equation $\mathcal{A}\mathbf{x}^2 = \rho \mathbf{x}^{[2]}$, where $\rho$ is the spectral radius of $\mathcal{A}$.
\end{theorem}

\begin{IEEEproof}
The proof relies on establishing an equivalence between the recursive definition of geometric capacity and the power iteration method used for finding the tensor's Perron vector.

Consider the standard power iteration algorithm for computing the spectral radius of the tensor $\mathcal{A}$. Let $\{\mathbf{y}^{(t)}\}$ be the sequence of vectors generated by the iteration, starting with $\mathbf{y}^{(0)} = \mathbf{1}$:
\begin{equation}\label{eq:power_iter_seq}
    y_i^{(t)} = \left(\sum_{j=1}^n \sum_{k=1}^n a_{ijk} y_j^{(t-1)} y_k^{(t-1)}\right)^{\frac{1}{2}}.
\end{equation}
We observe that this iterative update rule is mathematically identical to the recurrence formula for geometric capacity given in Eq.~\eqref{eq:geometric_recurrence}. To formally verify this, we prove that $\mathbf{C}_t = \mathbf{y}^{(t)}$ for all $t \ge 0$ via mathematical induction.

First, for the base case ($t=0$), by definition $C_0(i) = 1$ for all $i$, and the power iteration is initialized with $y_i^{(0)} = 1$. Thus, $\mathbf{C}_0 = \mathbf{y}^{(0)}$. Next, for the inductive step, assume $\mathbf{C}_{t-1} = \mathbf{y}^{(t-1)}$. Substituting this into the recursion yields:
\begin{equation}
\begin{split}
C_t(i) &= \sqrt{\sum_{j=1}^n \sum_{k=1}^n a_{ijk} C_{t-1}(j) C_{t-1}(k)} \\
       &= \sqrt{\sum_{j=1}^n \sum_{k=1}^n a_{ijk} y_j^{(t-1)} y_k^{(t-1)}} \\
       &= y_i^{(t)}.
\end{split}
\end{equation}
This establishes that the vector of geometric capacities $\mathbf{C}_t$ is identical to the iterative vector $\mathbf{y}^{(t)}$.

With this equivalence established, we use the Perron-Frobenius theory for nonnegative tensors \cite{Qi_2017}. Since $\mathcal{A}$ is a nonnegative weakly primitive tensor, the normalized sequence of the power iteration is guaranteed to converge to the unique positive eigenvector $\mathbf{x}$ associated with the spectral radius $\rho$. Consequently, the normalized capacity vector $\mathbf{u}_t$ satisfies:
\begin{equation*}
    \lim_{t \to \infty} \mathbf{u}_t = \lim_{t \to \infty} \frac{\mathbf{C}_{t}}{\|\mathbf{C}_t\|} = \lim_{t \to \infty} \frac{\mathbf{y}^{(t)}}{\|\mathbf{y}^{(t)}\|} = \mathbf{x},
\end{equation*}
where $\mathbf{x}$ is the unique HTEC vector satisfying $\mathcal{A}\mathbf{x}^2 = \rho \mathbf{x}^{[2]}$.
\end{IEEEproof}

This theorem provides a concrete combinatorial interpretation: the HTEC component $x_i$ represents the limit geometric capacity of the two-steps expansion tree rooted at node $i$, thereby measuring the ability of node $i$ to aggregate influence from its paired neighbors.

\subsection{Numerical illustration and algorithm}
The above subsection provided the theoretical interpretation. We now examine the calculation of HTEC using the power iteration method. This approach reveals the reinforcement mechanism where nodes and hyperedges mutually update their importance. To make it clear, we interpret the structural meaning of each iteration by detailing the process.

Given a connected hypergraph $H$ and its two-steps tensor $\mathcal{A} = (a_{ijk})$, the centrality vector is computed via the following iteration:

\begin{equation*}
x_i^{(t+1)} =  \left( \sum_{j,k=1}^n a_{ijk} x_j^{(t)} x_k^{(t)} \right)^{\frac{1}{2}}.
\end{equation*}
Starting from an initial all-one vector $\mathbf{x}^{(0)}=\mathbf{1}$, this non-linear update reinforces the node scores based on the joint influence of two-steps walks.  Since $\mathcal{A}$ is weakly primitive, this iterative sequence is guaranteed to converge to the unique positive Perron vector.

We describe the details of the first two iterations. To start with an initial vector $\mathbf{x}^{(0)} = \mathbf{1}$, which means we assume every node and hyperedge has the same initial importance. When we substitute $\mathbf{x}^{(0)} = \mathbf{1}$ into the update formula, we have
\begin{equation*}
  x_i^{(1)} =  \sqrt{\sum_{j,k=1}^{n} a_{ijk} x_j^{(0)}x_k^{(0)}}
 = \sqrt{\sum_{j,k=1}^{n} a_{ijk} \cdot 1 \cdot 1} = \sqrt{\sum_{j,k=1}^{n} a_{ijk}}.
\end{equation*}
Depending on whether $i$ is a node or a hyperedge, this sum has a structural meaning. 

\textbf{Case 1: }For a node $v$, the value $x_v^{(1)}$ is determined by the sum of the sizes of all hyperedges that node $v$ belongs to (mathematically, $\sum_{v \in e} |e|$). In other words, in the first step, a node gets a higher score if it joins larger groups. It reflects the number of neighbors a node can reach. 

\textbf{Case 2: }For a hyperedge $e$: the value $x_e^{(1)}$ is determined by the sum of the degrees of all nodes contained in hyperedge $e$ (mathematically, $\sum_{v \in e} d(v)$). That is, in the first step, a hyperedge gets a higher score if it contains nodes that belong to many other groups.

At the second iteration, we have
\begin{equation*}
  x_i^{(2)} =  \sqrt{\sum_{j,k=1}^{n} a_{ijk} x_j^{(1)}x_k^{(1)}}.
\end{equation*}
We analyze the specific meaning for nodes and hyperedges. Specifically, for a node $v \in V$:
$$x_v^{(2)} = \sqrt{\sum_{v,u \in e} x_e^{(1)} x_u^{(1)}}.$$
The centrality of node $v$ is updated by combining the importance of the hyperedge $e$ and the node $u$. If either the hyperedge $e$ is unimportant or the neighbor $u$ is weak, the product drops. Therefore, it indicates that a node achieves a high score only if it connects to other highly ranked nodes and hyperedges. Symmetrically, the centrality of hyperedge $e$ is updated using the scores of the node $v$ and the intersecting hyperedge $f$:
$$x_e^{(2)} = \sqrt{\sum_{v \in e,v \in f} x_v^{(1)} x_f^{(1)}}.$$
This formula implies that a hyperedge achieves a high score only if it contains highly ranked nodes $x_v^{(1)}$ that are also part of other highly ranked hyperedges $x_f^{(1)}$. Consequently, the method ranks a hyperedge based on the significance of its intersections, rather than simply its size or the number of nodes it contains.

Based on the above discussion, we propose Algorithm 1 for calculating the hypergraph two-steps eigenvector centrality. Algorithm \ref{alg:hyper_centrality} involves computing the spectral radius and the corresponding eigenvector of a tensor.
Those interested in this topic may refer to Refs. \cite{Ng_2010,Liu_2010,Zhou_2013}.

\begin{algorithm}[htbp]
    \caption{The Hypergraphs two-steps Eigenvector Centrality(HTEC)}
    \label{alg:hyper_centrality}
    \begin{algorithmic}[1] 
    \setlength{\itemsep}{3pt}
        \State \textbf{Inputs:} A connected hypergraph $H = (V, \mathcal{E})$ with $n_v = |V|$ nodes and $n_e = |\mathcal{E}|$ hyperedges.
        \State \textbf{Outputs:} 
        \State \qquad Node centrality vector $\mathbf{x}_V \in \mathbb{R}^{n_v}$;
        \State \qquad Hyperedge centrality vector $\mathbf{x}_E \in \mathbb{R}^{n_e}$.
        
        \State \textbf{Step 1:} Construct the incidence bipartite graph $B(H)$ where the node set is $V \cup \mathcal{E}$ (first $n_v$ indices for nodes, subsequent $n_e$ indices for hyperedges). Total dimension $n = n_v + n_e$.
        
        \State \textbf{Step 2 :} Construct the two-steps tensor $\mathcal{A}$ from $B(H)$.
        
        \State \textbf{Step 3:} Choose an initial positive vector $\mathbf{x}^{(0)} \in \mathbb{R}^n$ (e.g., $\mathbf{x}^{(0)} = \mathbf{1}$). Calculate $\mathbf{y}^{(0)}=\mathcal{A} (\mathbf{x}^{(0)})^{2}$. Set iteration counter $k=1$.
        
        \State \textbf{Step 4:} Compute the updated vector and bounds:
        \begin{equation*}
        \begin{split}
            \mathbf{x}^{(k)} &= \frac{(\mathbf{y}^{(k-1)})^{[\frac{1}{2}]}}{\| (\mathbf{y}^{(k-1)})^{[\frac{1}{2}]} \|}, \quad \mathbf{y}^{(k)}=\mathcal{A} (\mathbf{x}^{(k)})^{2}, \\
            \underline{\lambda}_k &= \min_{i} \frac{y^{(k)}_i}{(x^{(k)}_i)^{2}}, \quad \overline{\lambda}_k = \max_{i} \frac{y^{(k)}_i}{(x^{(k)}_i)^{2}}.
        \end{split}
        \end{equation*}

        \State \textbf{Step 5:} 
        \If {$\overline{\lambda}_k = \underline{\lambda}_k $}
            \State Go to \textbf{Step 6}.
        \Else
            \State Set $k \leftarrow k+1$ and return to \textbf{Step 4}.
        \EndIf

        \State \textbf{Step 6 :} Partition the converged vector $\mathbf{x}^{(k)}$ to obtain the final centralities:
        \[
        \mathbf{x}_V = \left( x^{(k)}_1, \dots, x^{(k)}_{n_v} \right)^\top, \quad \mathbf{x}_E = \left( x^{(k)}_{n_v+1}, \dots, x^{(k)}_{n} \right)^\top.
        \]
    \end{algorithmic}
\end{algorithm}

To demonstrate the mechanism of the proposed hypergraph two-steps centrality, we construct a non-uniform sunflower hypergraph, as visualized in Fig. \ref{sunflower}.  The sunflower hypergraph is constituted by a single central hub, node $v_1$, which is shared by six distinct hyperedges $\{e_1, \dots, e_6\}$. These hyperedges are designed with increasing cardinalities, ranging from a size of $|e_1|=2$ to $|e_6|=7$. Except for the central hub $v_1$, all other nodes are leaf nodes that appear in only one hyperedge.
\begin{figure}[!t]
\centering
\includegraphics[width=58mm]{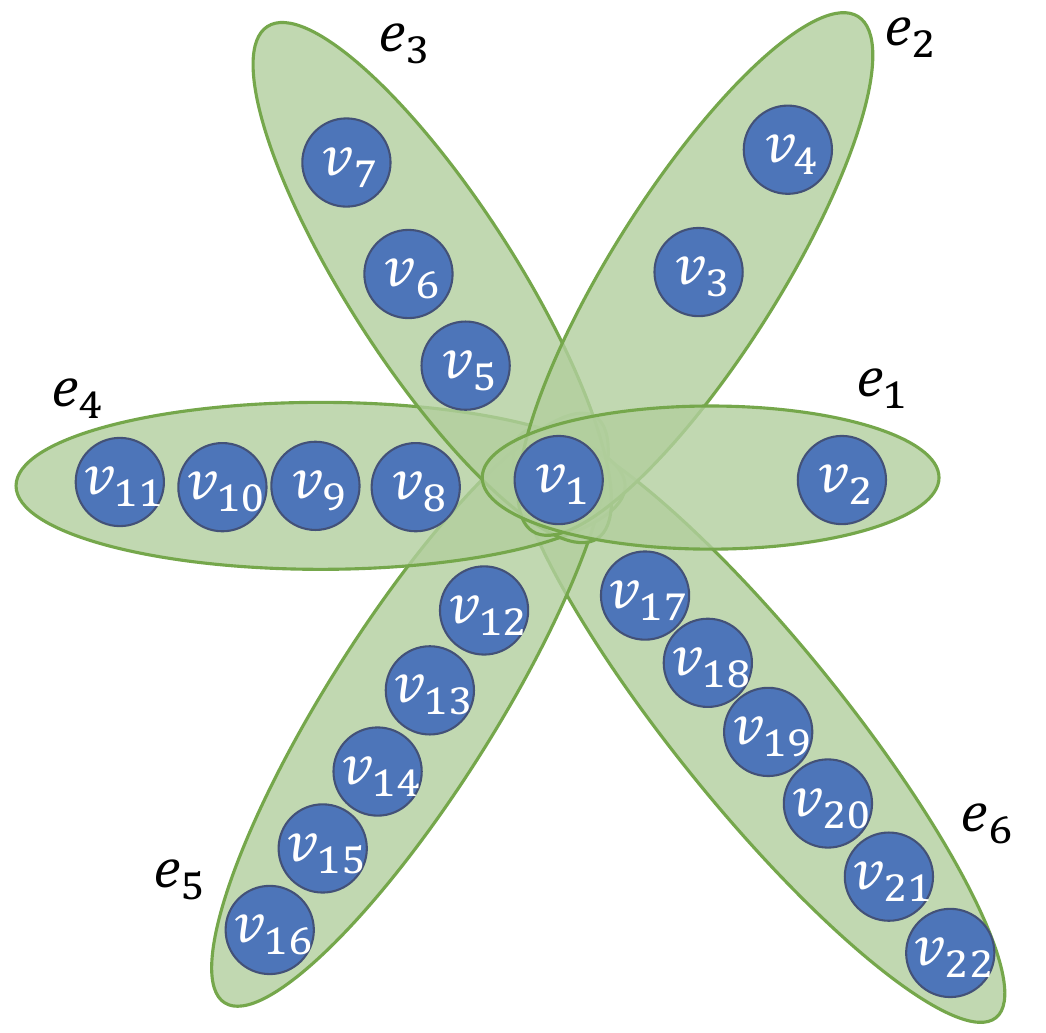}
\caption{A non-uniform sunflower hypergraph with hyperedge size from 2 to 7.}
\label{sunflower}
\end{figure}

We compute the hypergraph two-steps centrality vector $\mathbf{x}$ and summarize the numerical results in Table \ref{tab:hyperstar_results}. The analysis reveals a cardinality reinforcement mechanism. First, as expected, the central hub $v_1$ achieves the highest centrality score ($0.3489$) due to its global connectivity across all groups. While all leaf nodes have a degree of 1, their scores differ significantly based on the size of the hyperedge they belong to. Specifically,  node $v_2$ in the smallest hyperedge ($|e_1|=2$) receives the lowest score ($0.0941$), whereas nodes $v_{17}- v_{22}$ in the largest edge ($|e_6|=7$) achieve the highest leaf scores ($0.1953$). Similarly, the hyperedges themselves are ranked by size, with the largest hyperedge $e_6$ identified as the most significant ($0.2749$) and the smallest hyperedge $e_1$ ranked lowest ($0.2192$). The mechanism is straightforward: larger groups provide more connections. In our tensor model, every member in a hyperedge contributes to a node's score via a two-steps walk ($v_i \to e \to v_j$). For instance, in the largest edge ($|e|=7$), a node receives contributions from 6 other members; in contrast, in the smallest edge ($|e|=2$), it receives contribution from only 1 other member. Since the algorithm sums up these contributions, nodes in larger collaborations naturally accumulate higher scores.

\begin{table}[htbp]
\centering
\caption{Centrality scores in the non-uniform sunflower hypergraph.}
\label{tab:hyperstar_results}
\begin{tabular}{cccc}
\toprule
Edge ID &  Edge Score ($x_e$) & Node ID & Node Score ($x_{v}$) \\
\midrule 
-&- &$v_1$ &\textbf{0.3489} \\
$e_1$  & 0.2192 &$v_2$ & 0.0941 \\
$e_2$  & 0.2249 &$v_3,v_4$ & 0.1076 \\
$e_3$  & 0.2324 &$v_5-v_7$ & 0.1235 \\
$e_4$  & 0.2425 &$v_8-v_{11}$ & 0.1426 \\
$e_5$  & 0.2560 &$v_{12}-v_{16}$ & 0.1659 \\
$e_6$  & \textbf{0.2749} &$v_{17}-v_{22}$ & 0.1953 \\
\bottomrule
\end{tabular}
\end{table}

\section{Experiments on real-world hypergraphs}

In this section, we analyze the proposed hypergraph two-steps eigenvector centrality model (HTEC) on two real-world hypergraphs: \textit{Math-StackExchange co-tags} dataset \cite{Benson_2013}  and \textit{Walmart-Trips} dataset \cite{Amburg_2020}, which are available at https://github.com/ftudisco/node-edge-hypergraph-centrality . The Math-StackExchange co-tags dataset captures the higher-order co-tagging structure of mathematical questions on math.stackexchange.com. Here, nodes represent tags, and hyperedges  are formed by sets of tags applied to the same question. The Walmart-Trips dataset is derived from transactional data describing customer purchasing behavior. In this hypergraph, nodes represent specific products, while each hyperedge corresponds to a set of items co-purchased in a single trip. The basic statistics of these two datasets are summarized in Table \ref{tab:datasets}.

\begin{table*}[!t]
  \centering
  \caption{Statistics of Math-StackExchange co-tags and Walmart-Trips datasets.}
  \label{tab:datasets}
  \begin{tabular}{lcccc}
    \toprule
    Dataset & \makecell{Number of \\ nodes} & \makecell{Number of \\ hyperedges} & \makecell{Avg. hyperedge \\ cardinality} & \makecell{Max. hyperedge \\ cardinality} \\
    \midrule
    Math-StackExchange & 1,629 & 170,476 & 3.48 & 5 \\
    Walmart-Trips    & 88,860 & 65,979  & 6.86 & 25 \\
    \bottomrule
  \end{tabular}
\end{table*}
 
To analyze the behavior of the proposed hypergraph two-steps eigenvector centrality (HTEC), we compare its ranking results with three hypergraph centrality models: Linear, Max, and Log-Exp model centrality \cite{Tudisco_2021}. These centrality measures will be introduced later. Fig. 4 shows the scatter plots of node and hyperedge centrality on the Walmart and Math-StackExchange datasets, respectively. The x-axis represents the proposed HTEC, while the y-axis represents the other centrality measures. Each dot represents an individual node or hyperedge.

\begin{figure*}[!t]
\centering
    \subfloat[Math-StackExchange co-tags dataset node centrality\label{fig2:sub1}]{%
       \includegraphics[width=0.85\textwidth]{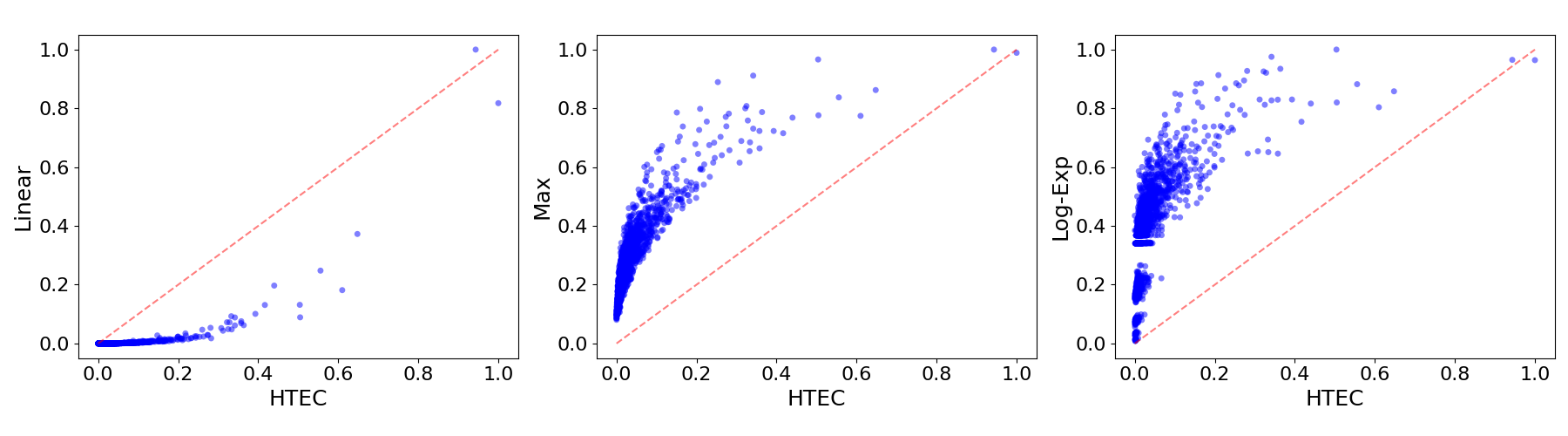}
    }
    \hfil 
    \subfloat[Math-StackExchange co-tags dataset hyperedges centrality\label{fig2:sub2}]{%
       \includegraphics[width=0.85\textwidth]{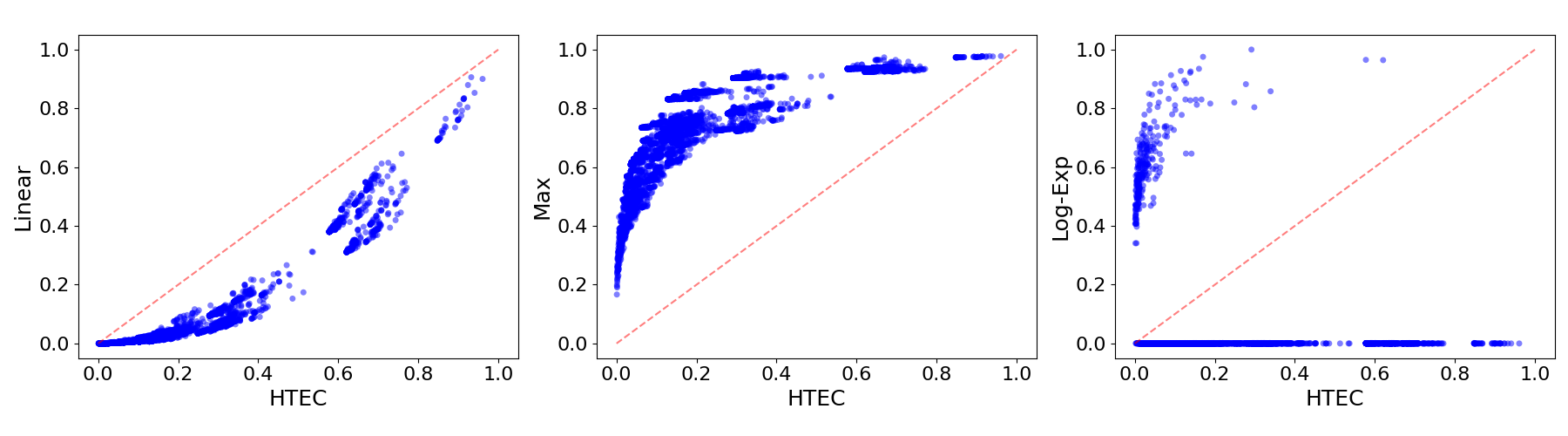}
    }
    \\ 
    
    \subfloat[Walmart trips dataset node centrality\label{fig2:sub3}]{%
       \includegraphics[width=0.85\textwidth]{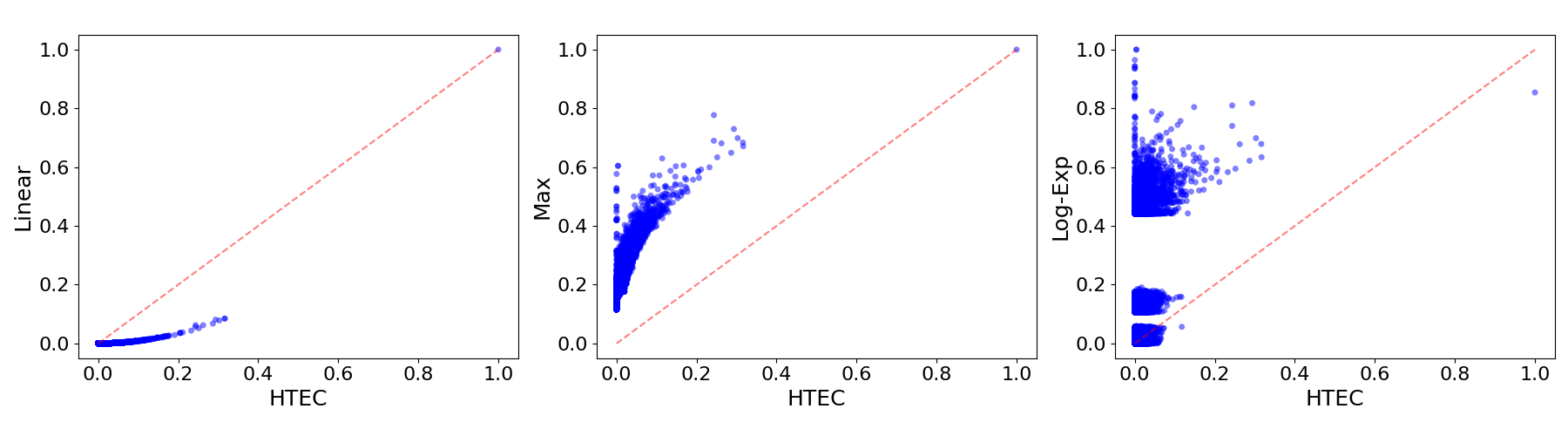}
    }
    \hfil
    \subfloat[Walmart trips dataset hyperedge centrality\label{fig2:sub4}]{%
       \includegraphics[width=0.85\textwidth]{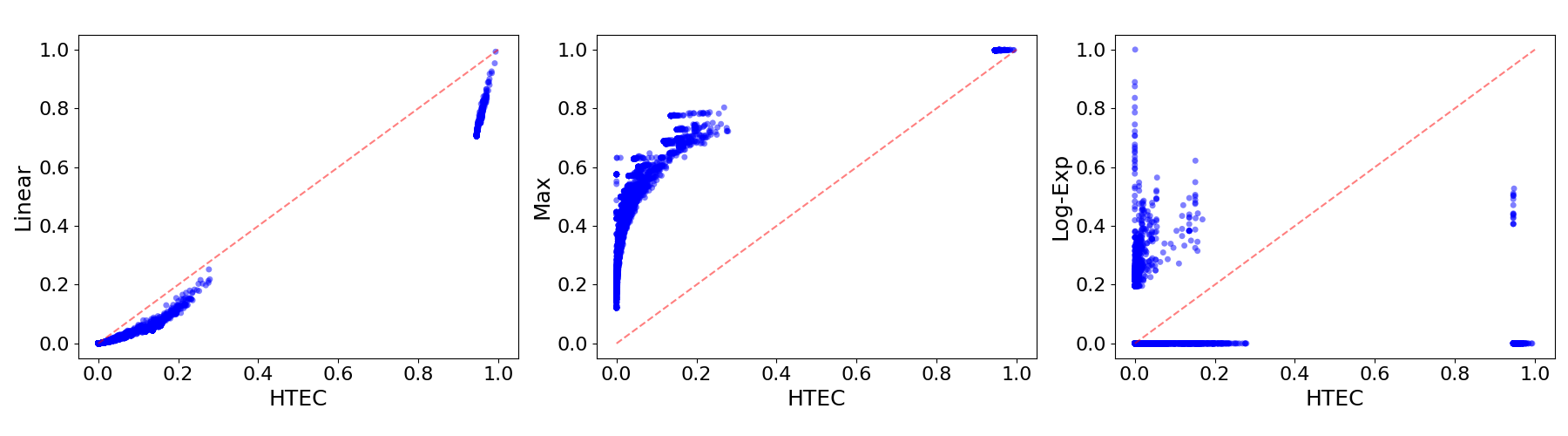}
    }
    
    \caption{Scatter plots of node and edge centralities by HTEC, Linear, Max and Log-exp model centralities.}
    \label{fig:all2}
\end{figure*}

The Linear centrality model corresponds to the standard eigenvector centrality applied to the graph and the line graph obtained by clique expanding the input hypergraph \cite{Tudisco_2021}. Comparing this with our HTEC method shows different patterns for hyperedges and nodes, see the left column of Fig. 4. For hyperedges, the data points align closely with the diagonal, indicating a high consistency between the two methods in ranking hyperedges. In contrast, the node centrality plot displays a more dispersed distribution. Note that in the Walmart trips hypergraph, there is a visible cluster of nodes that receive low scores in the Linear model but are assigned higher values in the HTEC model. 

In Max centrality model, a node's importance is not determined by the sum of its connections, but primarily by the single most important hyperedge it belongs to \cite{Tudisco_2021}. The comparison with the Max centrality model reveals a strong consistency in how the two methods rank importance. As shown in the middle column of Fig. 4, the data points indicate that both models identify the nodes in a similar order. Since the Max model assigns a node's importance based on its single most important incident hyperedge, this similarity indicates that the proposed HTEC method is also significantly influenced by the most prominent hyperedges a node belongs to. This behavior is consistent with the mutual reinforcement mechanism, where a node's score is elevated by its association with influential groups.

The Log-Exp centrality model generalizes Z-eigenvector centrality of uniform hypergraphs to non-uniform hypergraphs using logarithmic and exponential mappings \cite{Tudisco_2021}. The relationship between the HTEC and the Log-Exp model shows the weakest correlation among all methods. As observed in the right column of Fig. 4, many nodes have low Log-Exp scores but varying HTEC scores. This happens because the Log-Exp model multiplies the scores of neighbors. In this process, a single low-ranked node reduces the entire hyperedge score, even if the other nodes are important. In contrast, the HTEC method calculates scores by adding up the two-steps walks. This means that the final score reflects the total contribution of all neighbors. Therefore, even if a hyperedge contains a low-ranked node, the presence of other high-ranked nodes allows the hyperedge to maintain a high score.

To provide a concrete example of the ranking results, Table \ref{tab:math_results} lists the top-10 tags identified by each method. The results for HTEC, Linear, and Max are highly consistent. All three methods place fundamental topics like ``Calculus'' and ``Real analysis'' in the top positions. This confirms that HTEC, similar to the other centrality measures, effectively identifies the most important nodes. In contrast, the Log-Exp model produces a different list. For instance, it ranks ``Probability'' much higher (2nd place) and includes unique tags such as ``Algebra precalculus'' that do not appear in the other lists. 

\begin{table}[!t]
\centering
\caption{Top 10 nodes in the math stackexchange co-tags hypergraph.}
\label{tab:math_results}
\resizebox{\columnwidth}{!}{%
\begin{tabular}{llll}
\toprule
Linear \cite{Tudisco_2021} & Max \cite{Tudisco_2021} & Log-exp \cite{Tudisco_2021} & HTEC \\
\midrule 
Calculus & Calculus & Linear algebra & Real-analysis \\
Real analysis & Real analysis & Probability & Calculus \\
Integration & Linear algebra & Calculus & Integration \\
Sequences and series & Probability & Real analysis & Analysis \\
Limits & Abstract algebra & Complex analysis & Sequences-and-series \\
Analysis & Integration & Algebra precalculus & Functional-analysis \\
Derivatives & Sequences and series & General topology & Linear-algebra \\
Linear algebra & Matrices & Differential equations & Limits \\
Multivariable calculus & General topology & Combinatorics & Derivatives \\
Definite integrals & Combinatorics & Geometry & Multivariable-calculus \\
\bottomrule
\end{tabular}%
}
\end{table}

Next, we analyze the consistency of the top-$k$ rankings. Since real-world applications often prioritize the top-ranked elements over the entire list, we calculate Kendall and Spearman correlation coefficients between the HTEC method and the three other methods: Linear, Max, and Log-Exp centrality model \cite{Tudisco_2021}. The results for the Walmart trips and Math-StackExchange cotags datasets are shown in Fig. 5.

\begin{figure*}[!t]
\centering
    \subfloat[Math-StackExchange co-tags dataset node centrality correlations\label{fig3:sub1}]{%
       \includegraphics[width=0.85\textwidth]{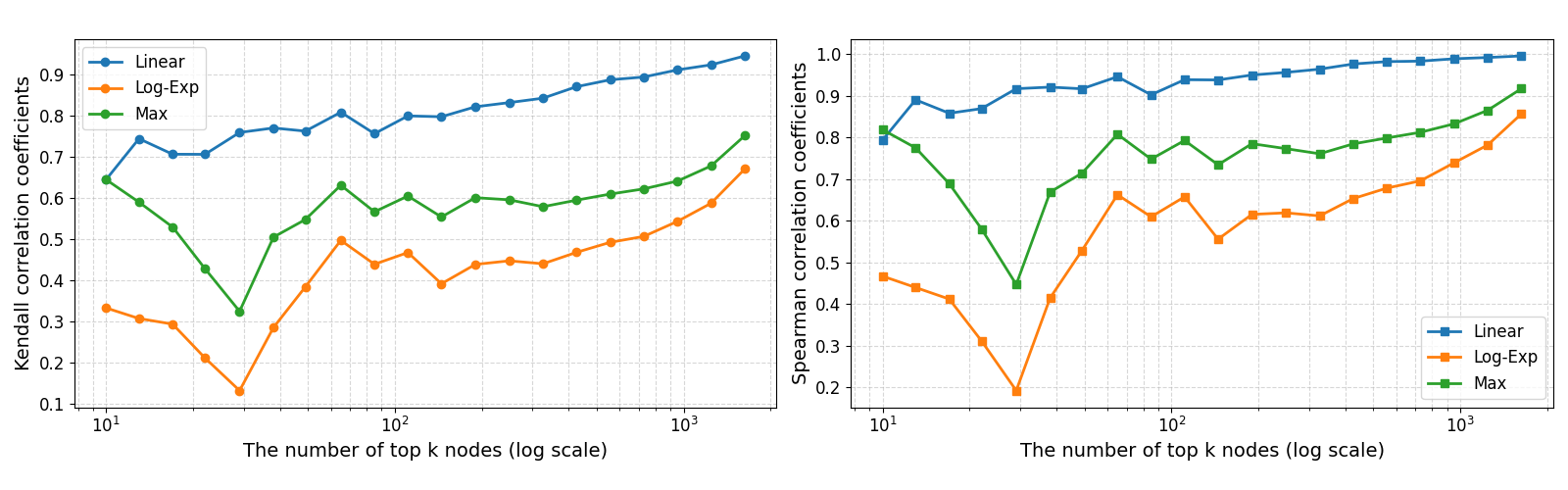}
    }
    \hfil 
    \subfloat[Math-StackExchange co-tags dataset hyperedges centrality correlations\label{fig3:sub2}]{%
       \includegraphics[width=0.85\textwidth]{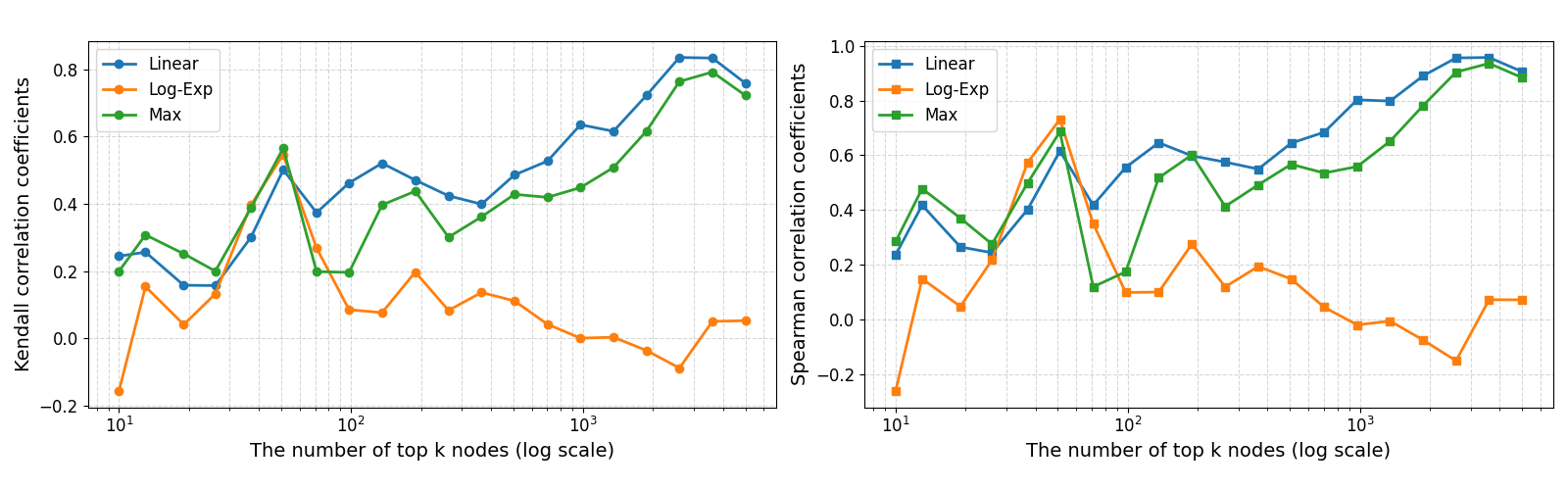}
    }
    \\ 

    \subfloat[Walmart trips dataset node centrality correlations\label{fig3:sub3}]{%
       \includegraphics[width=0.85\textwidth]{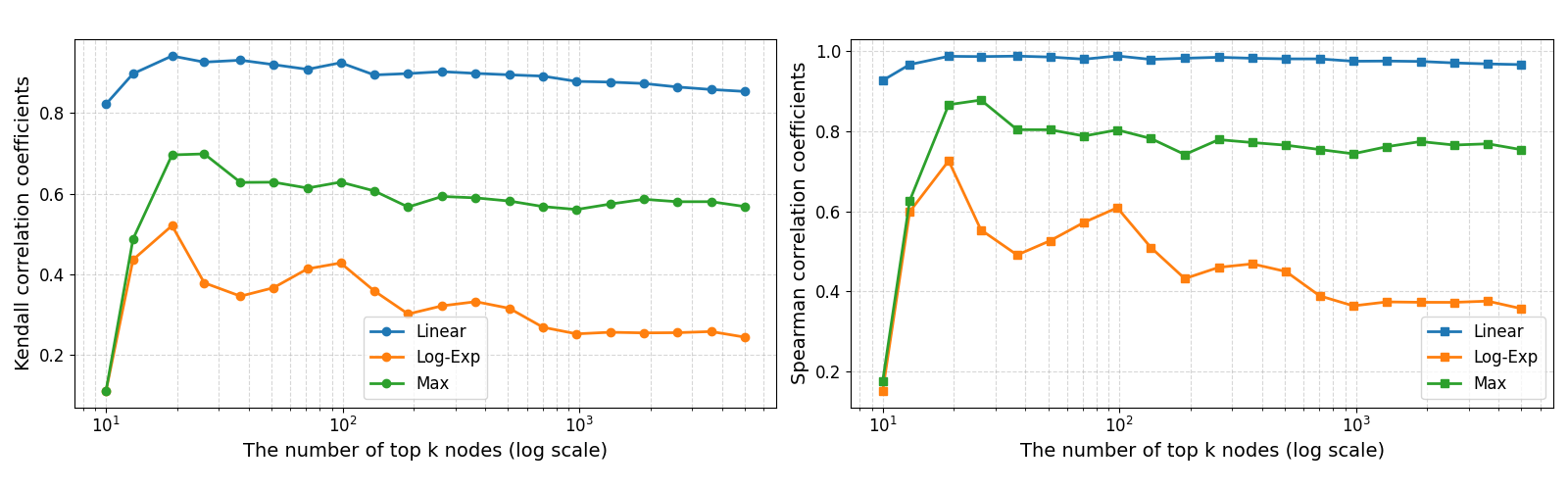}
    }
    \hfil 
    \subfloat[Walmart trips dataset hyperedges centrality correlations\label{fig3:sub4}]{%
       \includegraphics[width=0.85\textwidth]{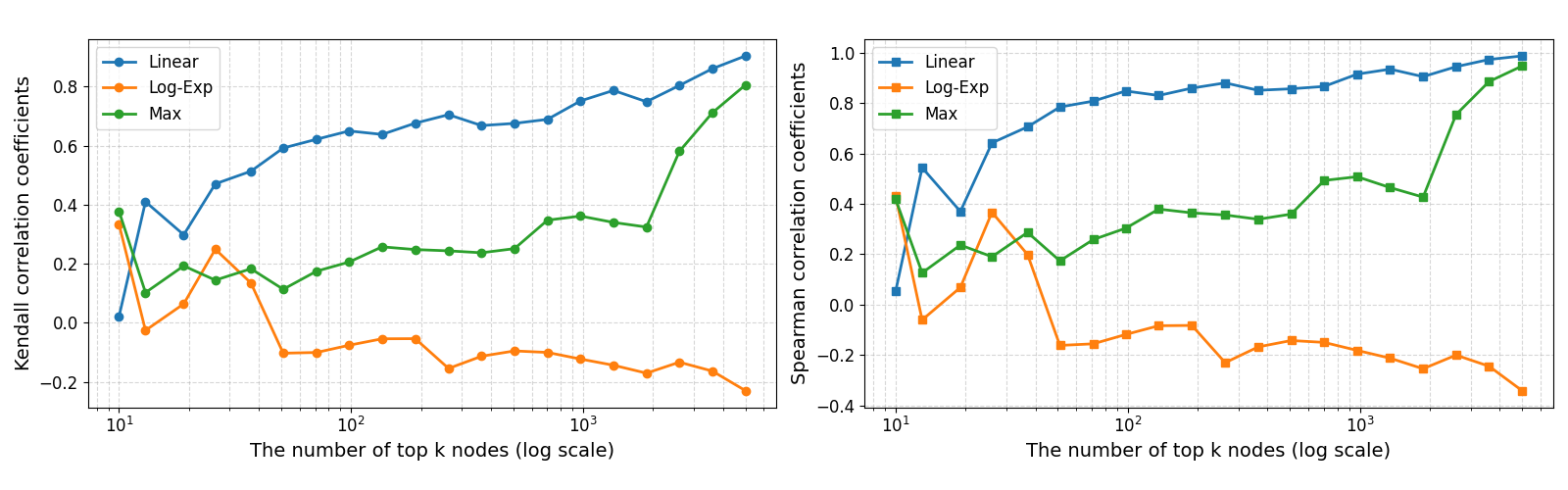}
    }

    \caption{Similarity of top-k ranked nodes and hyperedges by HTEC, Linear, Max and Log-exp model centralities.}
    \label{fig:all3}
\end{figure*}

For node centrality, the correlation trends in Fig. 5(a) and (c) show that the Linear model has the highest correlation with HTEC across all $k$ values (blue line). This strong alignment suggests that HTEC shares a key characteristic with the Linear model: both methods assign high rankings to nodes that have many connections. In contrast, the Log-Exp model (orange line) displays the weakest correlation, positioning significantly below the others. This divergence stems from its specific ``product-then-sum'' mechanism: unlike HTEC, the Log-Exp model first multiplies the scores of nodes within each hyperedge, meaning a single low-scoring node acts as a strict penalty that drastically reduces the hyperedge's contribution. Finally, the Max model (green line) falls in the middle. The reason is that Max centrality model only looks at the single best connection and ignores the rest.

In terms of hyperedge centrality, the correlation curves display a more complex pattern where the lines often intersect. At the beginning (when $k$ is small), the correlation values are low. This indicates that for the top-ranked nodes, each method identifies a different set of hyperedges. As we include more hyperedges (as $k$ gets larger), the correlation with the Linear model increases and eventually becomes higher than the others. This rising trend indicates that while the specific hyperedges at the very top of the lists differ between HTEC and Linear, the overall sets of selected hyperedges become increasingly consistent as the list grows. In other words, the disagreement is concentrated in the top rankings, whereas for larger lists, the two methods tend to include similar groups of hyperedges.

\section{Conclusion}
In this paper, we present a novel centrality measure for general hypergraphs, the Hypergraph Two-steps Eigenvector Centrality (HTEC). By transforming the hypergraph into its incidence bipartite graph and constructing a third-order tensor to encode two-steps walks, we derive a centrality vector that is guaranteed to exist and be unique via the Perron-Frobenius theorem. Furthermore, we provided a combinatorial interpretation for this measure: we proved that the centrality vector represents the limit geometric capacity of two-steps expansion trees, which measures the structural importance of nodes. The proposed HTEC measure captures a higher-order mutual reinforcement mechanism via a non-linear update rule. Unlike standard eigenvector centrality, which linearly sums neighbor scores, HTEC requires simultaneous strength from the connecting hyperedges and the neighbor nodes. Experimental results on two real-world hypergraphs confirmed that HTEC is able to identify important nodes and hyperedges, providing an effective tool for structural analysis in general hypergraphs. Future work will focus on extending the proposed tensor framework to multilayer networks, to quantify node centrality by combining different layers of connections.

\section*{Acknowledgement}

This work is supported by the National Natural Science Foundation
of China (No. 12071097, 12371344), the Natural Science Foundation for The Excellent Youth
Scholars of the Heilongjiang Province (No. YQ2024A009), and the Fundamental Research Funds for the Central
Universities (No. 2412025QD038).


\newpage

%
%
%
%

\vfill

\end{document}